\documentclass[a4paper,11pt]{article}
\usepackage[sc]{mathpazo}
\usepackage[T1]{fontenc}
\usepackage{amsthm,amsmath,amssymb,amsfonts,hyperref,a4wide}
\usepackage[usenames,dvipsnames]{xcolor} 
\usepackage{tikz}
\makeatletter
\tikzset{my loop/.style =  {to path={
  \pgfextra{}
  [looseness=12,min distance=6mm]
  \tikz@to@curve@path},font=\sffamily\small
  }}  
\makeatletter 
\usetikzlibrary{matrix,arrows,backgrounds,positioning,fit,shapes}

\newtheorem{theorem}{Theorem}[section]
\newtheorem{lemma}[theorem]{Lemma}

\newtheorem{proposition}[theorem]{Proposition}
\theoremstyle{definition}

\newtheorem{df}{Definition}[section]
\theoremstyle{remark}
\newtheorem{rem}{Remark}[section]
\newtheorem{ques}{Question}

\newcommand*{\Z}{\mathbb{Z}}

\newcommand*{\N}{\mathbb{N}}
\newcommand*{\C}{\mathbb{C}}

\newcommand*{\G}{\mathcal{G}}

\newcommand*{\ind}{\text{ind}}

\title{Computing the number of induced copies of a fixed graph in a bounded degree graph}
\author{ 
Viresh Patel\footnote{Korteweg de Vries Institute for Mathematics, University of Amsterdam. Email: \texttt{vpatel@uva.nl}. Supported by the Netherlands Organisation for Scientific Research (NWO) through the Gravitation Programme Networks (024.002.003).} \and Guus Regts\footnote{Korteweg de Vries Institute for Mathematics, University of Amsterdam. Email: \texttt{guusregts@gmail.com}. Supported by a personal NWO Veni grant}.}
\begin{document}
\maketitle
\abstract{In this paper we show that for any graph $H$ of order $m$ and any graph $G$ of order $n$  and maximum degree $\Delta$ one can compute the number of subsets $S$ of $V(G)$ that induces a graph isomorphic to $H$ in time $O(c^m \cdot n )$ for some constant $c = c(\Delta) >0$. This is essentially best possible.

\quad \\
\noindent 
\begin{footnotesize}
Keywords: induced graph, computational counting, fixed parameter tractability.
\end{footnotesize}
}

\section{Introduction}
For two graphs $H$ and $G$ we denote by $\ind(H,G)$ the number of subsets of the vertex set of $G$ that induce a graph that is isomorphic to $H$. (We recall that two graphs $H=(V_H,E_H)$ and $G=(V_G,E_G)$ are said to be \emph{isomorphic} if there exists a bijection $f:V_H \rightarrow V_G$ such that for any $u,v \in V_H$, we have that $f(u)f(v) \in E_G$ if and only if $uv \in E_H$.) Throughout we take $G$ and $H$ to have $n$ and $m$ vertices respectively.

Understanding the numbers $\ind(H,G)$ for different choices of $H$ gives us much important information about $G$. Determining these induced subgraph counts is closely related to determining subgraph counts and homomorphism counts; these parameters play a central role in the theory of graph limits \cite{L}, and frequently appear in statistical physics see e.g. Section 2.2 in~\cite{SS5} and the references therein.

When $H$ and $G$ are both part of the input, computing $\ind(H,G)$ is clearly an $NP$-hard problem because it includes the problem of determining the size of a maximum clique in $G$.
When the graph $H$ is fixed, the brute-force algorithm takes time $O(n^m)$ and improvements have been made for certain choices of the pattern graph $H$~\cite{KKM00, NP85}.

It is natural to consider this problem from a fixed parameter tractability (FPT) perspective.
In general computing $\ind(H,G)$ when parameterizing by $m = |H|$ is $W[1]$-hard because even deciding whether $G$ contains an independent set of size $m$ is $W[1]$-hard \cite{DF95}. 
Curticapean, Dell and Marx \cite{CDM17} prove a number of interesting dichotomy results for $W[1]$-hardness using the treewidth of $H$ (and of a certain class of graphs obtained from $H$) as an additional parameter.

However, when the graph $G$ is of bounded degree, which is often of interest in statistical physics, the problem is no longer $W[1]$-hard.
Indeed, Curticapean, Dell, Fomin, Goldberg, and Lapinskas~\cite[Theorem 13]{CDFGL17} showed that for a graph $H$ on $m$ vertices and a bounded degree graph $G$ on $n$ vertices, $\ind(H,G)$ can be computed in time $O(nm^{O(m)})$, thus giving an FPT algorithm in terms of $m$.
In the present paper we go further and give an algorithm with essentially optimal running time. We assume the standard word-RAM machine model with logarithmic-sized words.
\begin{theorem}\label{thm:main}
There is an algorithm which, given an $n$-vertex graph $G$ of maximum degree at most $\Delta$, and an $m$-vertex graph $H$ computes $\ind(H,G)$ in time $\tilde{O}(n(4\Delta)^{2m}+2^{10m})$. (Here the $\tilde{O}$-notation means that we suppress polynomial factors in $m$.)
\end{theorem}
\begin{rem}
Theorem 13 in \cite{CDFGL17} in fact concerns vertex-coloured graphs $H$ and $G$. 
Our proof of Theorem~\ref{thm:main} also easily extends to the coloured setting. 
We discuss this in Section~\ref{sec:conclusion}.
\end{rem}
The running time here is essentially optimal under the exponential time hypothesis. Indeed, if we could find an algorithm with an improved running time $O(c^{o(m)}poly(n))$ (for some constant $c$ possibly dependent on $\Delta$), we could use it to determine the size of a maximum independent set in time $O(nc^{o(n)}poly(n))=O(c^{o(n)}poly(n))$, which is not possible (even in graphs of maximum degree $3$) under the exponential time hypothesis; see \cite[Lemma 2.1]{JS98}.

Note that our algorithm allows us to compute $\ind(H,G)$ in polynomial time in $|G|$, even when $|H|$ is logarithmic in $|G|$ (and $G$ has bounded degree). The special case of this when $H$ is an independent set was a crucial ingredient in our recent paper~\cite{PR16}, which uses the Taylor approximation method of Barvinok~\cite{B17} to give (amongst others) a fully polynomial time approximation algorithm for evaluating the independence polynomial for bounded degree graphs. Our present paper completes the running time complexity picture for computing $\ind(H,G)$ on bounded degree graphs $G$.

We add a few remarks to give further perspective on the problem. Note that computing $\ind(H,G)$ in time $\tilde{O}(\Delta^m n)$ is relatively straightforward for $G$ of bounded degree $\Delta$ when $H$ is {\it connected} (see Lemma~\ref{lem:connected count}). Thus the difficulty lies in graphs $H$ that have many components. Note also that Curticapean et al.~\cite{CDFGL17} use the fact that induced graph counts can be expressed in terms of homomorphism counts (see e.g. \cite{L}) and that homomorphism counts from $H$ to $G$ can be computed in time $\tilde{O}(\Delta^m n)$ in their FPT algorithm. However the limiting factor is the time cost of expressing induced graph counts in terms of homomorphism graph counts, which is significantly larger than the running time of our algorithm.

Both our approach and the approach in \cite{CDFGL17} crucially use the bounded degree assumption.
It would be very interesting to know if Theorem~\ref{thm:main} could be extended to graphs of average bounded degree such as for example planar graphs. 
\begin{ques}
For which class of graphs $\mathcal C$ does there exist a constant $c=c(\mathcal C)$ and an algorithm such that given an $n$-vertex graph $G\in \mathcal C$ and an $m$-vertex graph $H$, the algorithm computes $\ind(H,G)$ in time $O(c^{m} poly(n))$? 
\end{ques}

\quad \\
{\bf Organization} The remainder of the paper is devoted to proving Theorem~\ref{thm:main}. 
The main idea in our proof is to define a multivariate graph polynomial where the coefficients are certain induced graph counts; in particular $\ind(H,G)$ will be the coefficient of a monomial. 
We then want to use machinery from \cite{PR16} to compute coefficients of univariate evaluations of this polynomial.
However, we need to slightly modify the result from \cite{PR16}. This will be done in the next section,  and in Section~\ref{sec:proof main}, we combine certain univariate evaluations to compute the coefficients of the multivariate polynomial. 

\section{Computing coefficients of graph polynomials}\label{sec:coef}
An efficient way to compute the coefficients of a large class of (univariate) graph polynomials for bounded degree graphs was given in \cite{PR16}. 
We will need a small modification of this result, for which we will provide the details here.
We start with some definitions after which we state the main result of this section.

By $\G$ we denote the collection of all graphs and by $\G_k$ for $k\in \N$ we denote the collection of graphs with at most $k$ vertices. 
A graph invariant is a function $f:\G\to S$ for some set $S$ that takes the same value on isomorphic graphs.
A (univariate) \emph{graph polynomial} is a graph invariant $p:\G\to \C[z]$, where $\C[z]$ denotes the ring of polynomials in the variable $z$ over the field of complex numbers.
Call a graph invariant $f$ \emph{multiplicative} if $f(\emptyset)=1$ and  $f(G_1\cup G_2)=f(G_1)f(G_2)$ for all graphs $G_1,G_2$ (here $G_1\cup G_2$ denotes the disjoint union of the graphs $G_1$ and $G_2$).
We can now give the key definition and tool we need from \cite{PR16}. 
\begin{df}\label{df:bigcp}
Let $p$ be a multiplicative graph polynomial defined by
\begin{equation}\label{eq:mult graph pol}
p(G)(z):=\sum_{i=0}^{d(G)} e_{i}(G)z^{i}
\end{equation}
for each $G\in \G$ with $e_0(G)=1$, where $d(G)$ is the degree of the polynomial $p(G)$.
We call $p$ a \emph{bounded induced graph counting polynomial (BIGCP)} if there exists $\alpha\in \N$ and a non-decreasing sequence $\beta \in \N^\N$ such that the following two conditions are satisfied:
\begin{itemize}
\item[(i)] for every graph $G$, the coefficients $e_i$ satisfy
\begin{equation}\label{eq:ind coef}
e_i(G):=\sum_{H\in \G_{\alpha i}}\lambda_{H,i}\ind(H,G)
\end{equation}
for certain $\lambda_{H,i}\in \C$;
\item[(ii)] for each $i$ and $H\in \G_{\alpha i}$, the coefficient $\lambda_{H,i}$  can be computed in time $\beta_i$.
\end{itemize}
\end{df}

We have the following result for computing coefficients of BIGCPs.
\begin{theorem}\label{thm:compute coef}
Let $n, m, \Delta\in \N$ and let $p(\cdot)$ be a bounded induced graph counting polynomial with parameters $\alpha$ and $\beta$.
Then there is a deterministic $\tilde{O}(n (e\Delta)^{\alpha m}\beta_m4^{\alpha m})
$-time algorithm, which, 
given any $n$-vertex graph $G$ of maximum degree at most $\Delta$, computes the first $m$ coefficients $e_1(G), \ldots, e_m(G)$ of $p(G)$. (Here the $\tilde{O}$-notation means that we suppress polynomial factors in $m$.)
\end{theorem}

\begin{rem}
The algorithm in the theorem above only has access to the polynomial $p$ via condition (ii) in the definition of BIGCP, that is, it relies only on the algorithm which computes the complex numbers $\lambda_{H,i}$.
\end{rem}

Before we prove Theorem \ref{thm:compute coef} we will first gather some facts from \cite{PR16} about induced subgraph counts and the number of connected induced subgraphs of fixed size that occur in a graph.
Compared to \cite{PR16} we actually need to slightly sharpen the statements. 

\subsection{Induced subgraph counts}
Define $\ind(H,\cdot):\G\to \C$ by $G\mapsto \ind(H,G)$. 
So we view $\ind(H,\cdot)$ as a graph invariant.
We can take linear combinations and products of these invariants.
In particular, for two graphs $H_1,H_2$ we have
\begin{equation}\label{eq:product of ind}
\ind(H_1,\cdot)\cdot \ind(H_2,\cdot)=\sum_{H\in \G}c^H_{H_1,H_2}\ind(H,\cdot),
\end{equation}
where for a graph $H$, $c^H_{H_1,H_2}$ is the number of pairs of subsets of $V(H)$, $(U,T)$, such that $U \cup T = V(H)$ and $H[U]=H_1$ and $H[T]=H_2$.
In particular, given $H_1$ and $H_2$, $c^H_{H_1,H_2}$ is nonzero for only a finite number of graphs $H$.

In what follows we will often have to maintain a list $L$ of subsets $S$ of $[n]$ with $|S|\leq k$ (for some $k$) as well as some (complex) number $c_S$ associated to $S$.
We will use the standard word-RAM machine model with logarithmic-sized words. This means that given a set $S$ of size $k$, we have access to $c_S$ in $O(k)$ time. 
In particular, this also means we can determine whether $S$ is contained in our list in $O(k)$ time. 

The next lemma says that computing $\ind(H,G)$ is fixed parameter tractable (and moreover gives an essentially optimal running time) when $G$ has bounded degree and $H$ is connected.
\begin{lemma}\label{lem:connected count}
Let $H$ be a connected graph on $k$ vertices and let $\Delta\in \N$. Then
\begin{itemize}
\item[(i)] there is an $O(n\Delta^{k-1})$-time algorithm, which, given any $n$-vertex graph $G$ with maximum degree at most $\Delta$, checks whether $\ind(H,G)\neq 0$;
\item[(ii)] there is an $O(n \Delta^{k-1}k)$-time algorithm, which, given any $n$-vertex graph $G$ with maximum degree at most $\Delta$, computes the number $\ind(H,G)$.
\end{itemize}
\end{lemma}
Note that Lemma~\ref{lem:connected count} (i) enables us to test for graph isomorphism between bounded degree graphs when $|V(G)| = |V(H)|$.
\begin{proof}
We follow the proof from \cite{PR16}.
We assume that $V(G)=[n]$.
Let us list the vertices of $V(H)$, $v_1,\ldots,v_k$ in such a way that for $i\geq 1$ vertex $v_i$ has a neighbour among $v_1,\ldots,v_{i-1}$.
Then to embed $H$ into $G$ we first select a target vertex for $v_1$ and then given that we have embedded $v_1, \ldots, v_{i-1}$ with $i\geq 2$ there are at most $\Delta$ choices for where to embed $v_i$. After $k$ iterations, we have a total of at most $n\Delta^{k-1}$ potential ways to embed $H$ and each possibility is checked in the procedure above. Hence we determine if $\ind(H,G)$ is zero or not in $O(n\Delta^{k-1})$ time.

Throughout the procedure above we maintain a list $L$ that contains all sets $S$ such that $G[S]=H$ found thus far.
Each time we find a set $S\subset [n]$ such that $G[S]=H$ we check if it is contained in $L$.
If this is not the case we add $S$ to $L$ and we discard $S$ otherwise.
The length of the resulting list gives the value of $\ind(H,G)$.
\end{proof}

Next we consider how to enumerate all possible connected induced subgraphs of fixed size in a bounded degree graph.
We will need the following result of Borgs, Chayes, Kahn, and Lov\'asz \cite[Lemma 2.1]{BCKL12}:
\begin{lemma}\label{lem:graph count}
Let $G$ be a graph of maximum degree $\Delta$. Fix a vertex $v_0$ of $G$.
Then the number of connected induced subgraphs of $G$ with $k$ vertices containing the vertex $v_0$ is at most $\frac{(e\Delta)^{k-1}}{2}$.
\end{lemma}

As a consequence we can efficiently enumerate all connected induced subgraphs of logarithmic size that occur in a bounded degree
graph $G$.
\begin{lemma}\label{lem:enumerate}
There is a $O(nk^3(e\Delta)^{k})$-time algorithm which, given $k \in \mathbb{N}$ and an $n$-vertex graph $G$ on $[n]$ of maximum degree $\Delta$, outputs $\mathcal{T}_k$, the list of all $S \subseteq [n]$ satisfying $|S| \leq k$ and $G[S]$ connected.
\end{lemma}

\begin{proof}
We assume that $V(G)=[n]$.
By the previous result, we know that $|\mathcal{T}_k| \leq n(e \Delta)^{k-1}$ for all $k$.

We inductively construct $\mathcal{T}_k$. For $k=1$, $\mathcal{T}_k$ is clearly the set of singleton vertices and takes time $O(n)$ to output.

Given that we have found $\mathcal{T}_{k-1}$ we compute $\mathcal{T}_k$ as follows. 
We iteratively compute $\mathcal{T}_{k}$ by going over all $S\in \mathcal{T}_{k-1}$ going over all $v\in N_G(S)$ (the collection of vertices that are connected to an element of $S$)  and checking whether $S\cup\{v\}$ is already contained in $\mathcal{T}_{k}$ or not. 
We add it to $\mathcal{T}_k$ if it is not already contained in $\mathcal{T}_k$.
 
The set $N_G(S)$ has size at most $|S|\Delta \leq k\Delta$ and takes time 
$O(k\Delta)$ to find (assuming $G$ is given in adjacency list form). 
Therefore computing $\mathcal{T}_k$ takes time bounded by $O(|\mathcal{T}_{k-1}|k^2\Delta) = O(nk^2(e\Delta)^{k})$.
 
Starting from $\mathcal{T}_1$, we perform the above iteration $k$ times, requiring a total running time of $O(nk^3(e\Delta)^{k})$.
The  proof that $\mathcal{T}_k$ contains all the sets we desire is straightforward and can be found in \cite{PR16}. 
\end{proof}

We call a graph invariant $f:\G\to \C$ \emph{additive} if for each $G_1,G_2\in \G$ we have  $f(G_1\cup G_2)=f(G_1)+f(G_2)$.
The following lemma is a variation of a lemma due to Csikv\'ari and Frenkel \cite{CF12}; it is fundamental to our approach. See \cite{PR16} for a proof.
\begin{lemma}\label{lem:additivity}
Let $f:\G\to \C$ be a graph invariant given by $f(\cdot):=\sum_{H\in \G}a_H\ind(H,\cdot)$. 
Then $f$ is additive if and only if $a_H=0$ for all graphs $H$ that are disconnected. 
\end{lemma}

The next proposition is a variant of the Newton identities that relate the inverse power sums and the coefficients of a polynomial. We refer to \cite{PR16} for a proof.
\begin{proposition}
\label{pr:newton}
Let $p(z) = a_0 + \cdots + a_dz^d$ be a polynomial of degree $d$ with complex roots $\zeta_1, \ldots, \zeta_d$. Define $p_j:= \zeta_1^{-j} +  \cdots + \zeta_d^{-j}$. Then for each $k=1,2, \ldots$, we have 
\[
ka_k= -\sum_{i=0}^{k-1}a_ip_{k-i}.
\]
(Here we take $a_i=0$ if $i>d$.)
\end{proposition}

\subsection{Proof of Theorem~\ref{thm:compute coef}}
We follow the proof as given in \cite{PR16}, which we modify slightly at certain points.

Recall that $p(\cdot)$ is a bounded induced graph counting polynomial (BIGCP). Given an $n$-vertex graph $G$ with maximum degree at most $\Delta$, we must show how to compute the first $m$ coefficients of $p$. We will use $\tilde{O}$-notation throughout to mean that we suppress polynomial factors in $m$.
To reduce notation, let us write $p=p(G)$, $d=d(G)$ for the degree of $p$, and $e_i=e_i(G)$ for $i=0,\ldots,d$ for the  coefficients of $p$ (from (\ref{eq:mult graph pol})).  
We also write $p_k:= \zeta_1^{-k} + \cdots + \zeta_d^{-k}$, where $\zeta_1,\ldots,\zeta_d\in \C$ are the roots of the polynomial $p(G)$.

Noting $e_0=1$, Proposition~\ref{pr:newton} gives
\begin{equation}\label{eq:newton}
p_k= -ke_k - \sum_{i=1}^{k-1}e_ip_{k-i},
\end{equation}
for each $k=1, \ldots, d$. 

By \eqref{eq:ind coef}, for $i\geq 1$, the $e_i$ can be expressed as linear combinations of induced subgraph counts of graphs with at most $\alpha i$ vertices. Since $p_1=-e_1$, this implies that the same holds for $p_1$.
By induction, \eqref{eq:product of ind}, and \eqref{eq:newton} we have that for each $k$
\begin{equation}\label{eq:power to ind}
p_k=\sum_{H\in \G_{\alpha k}}a_{H,k}\ind(H,G),
\end{equation}
for certain, yet unknown, coefficients $a_{H,k}$.

Since $p$ is multiplicative, the inverse power sums are additive. Thus Lemma \ref{lem:additivity} implies that $a_{H,k}=0$ if $H$ is not connected. 
Denote by $\mathcal{C}_{i}(G)$ the set of connected graphs of order at most $i$ that occur as induced subgraphs in $G$.
Let us assume that $G$ has vertex set $[n]$.
Denote by $\mathcal{T}_{\leq \alpha k}(G)$ the list consisting of those sets $S\subseteq[n]$ of size at most $\alpha k$ that induce a connected graph in $G$.
This way we can rewrite \eqref{eq:power to ind} as follows:
\begin{equation}\label{eq:power support}
p_k=\sum_{H\in \mathcal{C}_{\alpha k}(G)}a_{H,k}\ind(H,G)=\sum_{S\in \mathcal{T}_{\leq \alpha k}(G)}a_{G[S],k}.
\end{equation}
The next lemma says that we can compute the coefficients $a_{S,k}:=a_{G[S],k}$ efficiently for $k=1,\ldots,m$.

\begin{lemma}\label{lemma:compute}
There is an $\tilde{O}(n (e\Delta)^{\alpha m}\beta_m4^{\alpha m})$-time algorithm, which given a BIGCP $p$ (with parameters $\alpha$ and $\beta$) and an $n$-vertex graph $G$ of maximum degree $\Delta$, computes and lists the coefficients $a_{S,k}$ in \eqref{eq:power support} for all $S\in \mathcal{T}_{\leq \alpha k}(G)$ and all $k=1,\ldots, m$.
\end{lemma}
\begin{proof}
We assume that the vertex set of $G$ is equal to $[n]$.
Using the algorithm of Lemma \ref{lem:enumerate}, we first compute the list $\mathcal{T}_{\leq \alpha k}$ consisting of all subsets $S$ of $V(G)$ such that $|S|\leq \alpha k$ and $G[S]$ is connected. 
This takes time bounded by 
\begin{equation}\label{eq:list}
O(n(\alpha m)^3 (e\Delta)^{\alpha m}) = \tilde{O}(n(e\Delta)^{\alpha m}).
\end{equation}
(Note that the algorithm in Lemma \ref{lem:enumerate} actually computes $\mathcal{T}_{\leq \alpha k}$ when it computes $\mathcal{T}_{\alpha m}$.)

To prove the lemma, let us fix $k\leq m$ and show how to compute the coefficients $a_{S,k}$, assuming that  we have already computed and listed the coefficients $a_{S',k'}$ for all $k'<k$ and $S'\in \mathcal{T}_{\leq \alpha k'}$.
Let us fix $S\in \mathcal{T}_{\leq \alpha k}$.
Let $H=G[S]$. 
By (\ref{eq:newton}), it suffices to compute the coefficient of $\ind(H,\cdot)$ in $p_{k-i}e_{i}$ for $i=1,\ldots,k$ (where we set $p_0=1)$.
By \eqref{eq:ind coef}, \eqref{eq:product of ind} and \eqref{eq:power to ind} we know that the coefficient of $\ind(H,\cdot)$ in $p_{k-i}e_{i}$ is given by
\begin{equation}\label{eq:coef e_{k-i}p_i}
\sum_{H_1,H_2} c^H_{H_1,H_2}a_{H_2,(k-i)}\lambda_{H_1,i}=\sum_{(U,T) : U \cup T = V(H)} a_{H[T], (k-i)}\lambda_{H[U], i} .
\end{equation}
As $|V(H)|\leq \alpha k$, the second sum in \eqref{eq:coef e_{k-i}p_i} is over at most $4^{\alpha k}=O(4^{\alpha m})$ pairs $(U,T)$.
For each such pair, we need  to compute $\lambda_{H[U], i}$ and $ a_{H[T],(k-i)}$.
We can compute $\lambda_{H[U],i}$ in time bounded by $O(\beta_i)=O(\beta_m)$ since $p$ is a BIGCP.
As $H[T]=G[T]$, to compute $a_{H[T],(k-i)}$ we just need to look up the coefficient $a_{T,k-i}$,  which takes time $O(k-i)$.

Together, all this implies that the coefficient of $\ind(H,\cdot)$ in $p_{k-i}e_{i}$ can be computed in time bounded by 
\begin{equation}\label{eq:compute ahk}
O( 4^{\alpha m}( \beta_m+ m))=\tilde{O}(\beta_m\cdot 4^{\alpha m}). 
\end{equation}
So the coefficient $a_{H,k}$ can be computed in the same time (since we suppress polynomial factors in $m$).
Thus all coefficients $a_{S,k}$ for $S\in \mathcal{T}_{\leq \alpha k}$ can be computed and listed in time bounded by $|\mathcal{T}_{\leq \alpha k}|$ multiplied by the expression \eqref{eq:compute ahk}, which is bounded by 
\begin{equation}\label{eq:compute ahk2}
\tilde{O}(n (e\Delta)^{\alpha m}\beta_m4^{\alpha m})
\end{equation}
by Lemma~\ref{lem:graph count}.

So the total running time is bounded by the time to compute the list $\mathcal{T}_{\leq \alpha m}$ (which is given by  \eqref{eq:list}) plus the time to compute the $a_{S,k}$ for $S \in\mathcal{T}_{\leq \alpha k}$ (which is given by \eqref{eq:compute ahk2}) for $k=1,\ldots,m$. 
This proves the lemma.
\end{proof}

To finish the proof of the theorem, we compute $p_k$ for each $k=1,\ldots, m$ by adding all the numbers $a_{S,k}$ over all $S\in \mathcal{T}_{\leq \alpha k}(G)$ using \eqref{eq:power support} (these numbers were computed in the previous lemma in time $\tilde{O}(n (e\Delta)^{\alpha m}\beta_m4^{\alpha m})
$). 
Doing this addition takes time 
\[O(m|\mathcal{T}_{\leq \alpha m}(G)|)=\tilde{O}(n(e\Delta)^{\alpha m} ).
\] 
Finally, knowing the $p_i$, we can inductively compute the $e_i$ for $i=1, \ldots, m$ using the relations (\ref{eq:newton}), in quadratic time in $m$. 
So we see that the total  running time for computing $e_1, \ldots, e_m$ is dominated by the computation of the $\alpha_{S,k}$ and is $\tilde{O}(n (e\Delta)^{\alpha m}\beta_m4^{\alpha m})$.
This proves the theorem.

\section{Proof of Theorem~\ref{thm:main}}\label{sec:proof main}
We first set up some notation before we state our key definition.
For a graph $H$ we write $H = i_1H_1 \cup \cdots \cup i_rH_r$ to mean that $H$ is the disjoint union of $i_1$ copies of $H_1$, $i_2$ copies of $H_2$ all the way to $i_r$ copies of $H_r$ for connected and pairwise non-isomorphic graphs $H_1, \ldots, H_r$. 
For vectors $\mu,\nu \in \Z_{\geq 0}^r$ we write $\mu \circ \nu$ for the vector in $\Z_{\geq 0}^r$ that is the pointwise or Hadamard product of $\mu$ and $\nu$. We write $\mu^\nu$ for the product $\mu_1^{\nu_1} \cdots \mu_r^{\nu_r} \in \Z_{\geq 0}$, while $\mu \cdot \nu$ denotes the usual scalar product. 
For a vector of variables $x=(x_1,\ldots,x_r)$ and $\mu\in \Z_{\geq 0}^r$ we define $x^\mu:=x_1^{\mu_1}\cdots x_r^{\mu_r}$.

For pairwise non-isomorphic and connected graphs $H_1,\ldots,H_r$ write $\underline{H} = (H_1, \ldots, H_r)$; we define the multivariate graph polynomial $Z_{\underline{H}}(G) \in \mathbb{Z}[x_1, \ldots, x_r]$ as follows. For a graph $G$ we let 
\begin{align*}
Z_{\underline{H}}(G;x) = \sum_{\gamma \in \Z_{\geq 0}^r}  \ind(\gamma \underline{H}, G) x^{\gamma \circ h},
\end{align*}
where $x= (x_1, \ldots, x_r)$, $\gamma = (\gamma_1, \ldots, \gamma_r)$, $h: = (|V(H_1)|, \ldots, |V(H_r)|)$ and where $\gamma \underline{H}$ denotes the graph $\gamma_1H_1 \cup \cdots \cup \gamma_rH_r$.

Computing $\ind(H,G)$ for any two graphs $H= i_1H_1 \cup \cdots \cup i_rH_r$ and $G$ can now be modelled as computing the  coefficient of the monomial $x_1^{i_1h_1}x_2^{i_2h_2}\cdots x_r^{i_rh_r}$ in $Z_{\underline{H}}(G;x)$.
Let us start by gathering some facts about the polynomial $Z_{\underline{H}}$.
\begin{proposition}
\label{pr:mult}
The polynomial $Z_{\underline{H}}$ is multiplicative, i.e., for any two graphs $G_1$ and $G_2$, 
$Z_{\underline{H}}(G_1 \cup G_2;x) = Z_{\underline{H}}(G_1;x) \cdot Z_{\underline{H}}(G_2;x)$. In particular, any evaluation of $Z_{\underline{H}}$ is also multiplicative.
\end{proposition}
\begin{proof}
Note first that every monomial in $Z_{\underline{H}}(G;x)$ is of the form $x^{\gamma \circ h}$ for some unique choice of $\gamma$.  
For notational convenience we write $s_\gamma(G):= \ind(\gamma \underline{H}, G)$.
Consider the coefficient of $x^{\gamma \circ h}$ in the polynomial $Z_{\underline{H}}(G_1;x) \cdot Z_{\underline{H}}(G_2;x)$. The coefficient is given by
\[ \sum_{\mu + \nu = \gamma} s_{\mu}(G_1) s_{\nu}(G_2),
\]
which counts precisely the number of copies of $\gamma_1H_1 \cup \cdots \cup \gamma_rH_r$ in $G_1 \cup G_2$, that is, $s_{\gamma}(G_1 \cup G_2)$, which is the coefficient of $x^{\gamma \circ h}$ in the polynomial $Z_{\underline{H}}(G_1 \cup G_2;x)$.
\end{proof}

Suppose $\mu \in \Z_{\geq 0}^r$ and let $z$ be a variable. Define the graph polynomial $Z_{\mu} = Z_{\mu,\underline{H}}(G) \in \Z[z]$ by 
\[
 Z_{\mu}(z) = Z_{\underline{H}}(G;(\mu_1z, \ldots, \mu_rz)) =: \sum_{i \geq 0} s_i(\mu)z^i;
\]
here the second equality defines the numbers $s_i(\mu) = s_i(\mu)(G)$. 
In particular, we know that
\begin{equation}\label{eq:combined}
s_i(\mu) = \sum_{\substack{\gamma \in \Z_{
\geq 0}^r \\\gamma \cdot h = i}} \mu^{\gamma}\ind(\gamma\underline{H},G).
\end{equation}  

\begin{proposition}
\label{pr:BIGCP}
Fix $\underline{H} = (H_1, \ldots, H_r)$ where the $H_i$ are pairwise non-isomorphic connected graphs each of maximum degree at most $\Delta$ and fix $\mu \in \Z_{\geq 0}^r$. 
Then  $Z_{\mu,\underline{H}}(G;z)$ is a BIGCP with parameters $\alpha = 1$ and $\beta_i = i^2r \Delta^{i-1}$. 
\end{proposition}
\begin{proof}
Since $Z_{\mu,\underline{H}}(G)$ is a particular evaluation of $Z_{\underline{H}}(G)$, we know by Proposition~\ref{pr:mult} that it is multiplicative.

The coefficient of $z^i$ in $Z_{\mu,\underline{H}}(G;z)$ is given by \eqref{eq:combined}.
Since $\gamma \underline{H}$ is a graph with exactly $\gamma \cdot h = i$ vertices, we can take $\alpha$ to be $1$ in the definition of BIGCP. 

For a given graph $F$, we must determine $\lambda_{F,i}$ in the definition of BIGCP and the time $\beta_i$ required to do this. Note that we may assume $|V(F)| = i$; otherwise $\lambda_{F,i} = 0$. If $|V(F)| = i$, we must test if $F$ is isomorphic to a graph of the form $\gamma \underline{H}$ with $\gamma \cdot h = i$ and if so we must output the value of $\lambda_{F,i}$ as $\mu^{\gamma}$ (this last step taking $i$ arithmetic operations).
To test if $F$ is isomorphic to a graph of the form $\gamma \underline{H}$, we test isomorphism of each component of $F$ against each of the graphs $H_1, \ldots, H_r$, which takes time at most $O(ir\Delta^{i-1})$ using Lemma~\ref{lem:connected count} at most $ir$ times. Thus the total time to compute $\lambda_{F,i}$ is at most $O(i^2r\Delta^{i-1})$.
\end{proof}

Now since $Z_{\mu,\underline{H}}(G;z)$ is a BIGCP, Theorem~\ref{thm:compute coef} allows us to compute the coefficients $s_i(\mu)$ in \eqref{eq:combined} with the desired running time. 
However the $s_i(\mu)$ are linear combinations of the numbers $\ind(\gamma \underline{H}, G)$, while we wish to compute one of these numbers in particular, say $\ind(\rho \underline{H}, G)$. By making careful choices of different $\mu$, we will obtain an invertible linear system whose solution will include the number $\ind(\rho \underline{H},G)$.
We will require Alon's Combinatorial Nullstellensatz~\cite{A99}, which we state here for the reader's convenience.
\begin{theorem}[\cite{A99}]
\label{thm:null}
Let $f(x_1, \ldots, x_n)$ be a polynomial of degree $d$ over a field $\mathbb{F}$. Suppose the coefficient of the monomial $x_1^{\mu_1}\cdots x_n^{\mu_n}$ in $f$ is nonzero and $\mu_1+ \cdots + \mu_n = d$. If $S_1, \ldots, S_n$ are finite subsets of $\mathbb{F}$ with $|S_i| \geq \mu_i + 1$ then there exists a point $x \in S_1 \times \cdots \times S_n$ for which $f(x) \not=0$.
\end{theorem}

Given a vector $h \in\N^r$, let us write $\mathcal{P}_{m,r,h}$ for the set of vectors $\gamma \in \mathbb{Z}_{\geq 0}^r$ such that $\gamma \cdot h = m$.
We note that, as the the number of elements in $\mathcal{P}_{m,r,h}$ is at most the number of monomials in $r$ variables of degree $m$, we have 
\begin{equation}\label{eq:bound P}
|\mathcal{P}_{m,r,h}|\leq \binom{m+r-1}{r-1}.
\end{equation}

\begin{lemma}
\label{le:invertsys}
Fix  $m,r \in \mathbb{N}$ and $h \in \mathbb{N}^r$, and 
let $\gamma_1, \ldots, \gamma_k$ be an enumeration of the 
elements in $\mathcal{P}_{m,r,h}$. Given a 
vector $\nu \in \mathbb{N}^r$, let us write $\nu^* \in \mathbb{N}^k$ for the vector $(\nu^{\gamma_i \circ h})_{i=1}^n\in \N^k$. 
In time $O(k^5 + k^2m e^{m})$, we can find vectors $\nu_1, \ldots, \nu_k\in \mathbb{N}^r$ such that $\nu_1^*, \ldots, \nu_k^*$ are linearly independent. 
\end{lemma}

\begin{proof}
For any vector $\nu$, let us write $\nu|_j$ to denote the vector consisting of the first $j$ components. Suppose we have found vectors $\nu_1, \ldots, \nu_{\ell-1}\in  \mathbb{N}^r$ such that the $(\ell-1) \times (\ell-1)$ matrix 
\[
M_{\ell-1} := (\nu_1^*|_{\ell-1}, \ldots, \nu_{\ell-1}^*|_{\ell-1})
\]
has non-zero determinant. We will show how to find $\nu_\ell \in \mathbb{N}^r$ such that the corresponding matrix $M_\ell$ has non-zero determinant. 
First consider the components of $\nu_\ell$ to be unknown variables $x_1, \ldots, x_r$ so that $\det(M_\ell)$ becomes a polynomial $P = P(x_1, \ldots, x_r)$ in the variables $x_1, \ldots, x_r$. In fact it is a homogeneous polynomial of degree $m$. 
Writing $x = (x_1, \ldots, x_r)$, we know that the coefficient of $x^{\gamma_\ell \circ h}$ is $\det(M_{\ell-1}) \not= 0$  (consider the determinant expansion of the matrix $M_\ell$ along the $\ell th$ column).
We must now find $\nu_\ell \in \mathbb{N}^r$ such that $\det(M_\ell) = P(\nu_{\ell,1},\ldots,\nu_{\ell,\ell}) \not = 0$, where $v_{\ell, i}$ is the $i$th component of $v_{\ell}$. 

Assume the components of $\gamma_\ell \in \mathbb{N}^r$ are $a_1, \ldots, a_r$. 
Applying Theorem~\ref{thm:null} to the monomial $x^{\gamma_\ell \circ h}$ and taking the sets $S_i = \{1, \ldots, a_ih_i + 1 \}$ for $i=1,\ldots,r$, we know there exists a vector $\nu_\ell \in S := S_1 \times \cdots \times S_r$ such that $P(\nu_{\ell,1},\ldots,\nu_{\ell,\ell}) \not=0$.  
Computing the polynomial $P$ requires time at most $O(k \cdot k^{3})$ (using that computing the determinant of an $n \times n$ matrix takes $O(n^3)$ time) and evaluating it at every point in $S$ requires at most  $O(m \cdot k \cdot |S|)$ operations. We can bound $|S|$ as follows:
\begin{align*}
|S| = (a_1h_1+1) \cdots (a_rh_r + 1) 
\leq \left( \frac{1}{r} \sum_{i=1}^r (a_ih_i+1) \right)^r 
\leq \left( \frac{m + r}{r} \right)^{r}  
= \left( 1 + \frac{m}{r} \right)^r 
\leq e^{m}.
\end{align*}
The first inequality follows from the arithmetic-geometric mean inequality. Iterating the procedure, we can determine $\nu_1, \ldots, \nu_k$ in time $O(k\cdot(k^4 + m  k |S|)) \leq O(k^5 + k^2m e^{m})$.
\end{proof}
\begin{rem}
We suspect there should be a simpler argument than the one we have just given (perhaps one where the vectors $\nu_1, \ldots, \nu_k$ can be explicitly written down rather than having an algorithm to determine them). Note that one can also use a faster randomised algorithm by applying the Schwarz-Zippel Lemma.
\end{rem}

We can now prove Theorem~\ref{thm:main}.

\begin{proof}[Proof of Theorem~\ref{thm:main}]
We may assume that $\Delta(H)\leq \Delta$, for otherwise $\ind(H,G)=0$.
Let $H=\rho_1H_1\cup \rho_2H_2\cup \ldots\cup \rho_rH_r$, where the $H_j$ are components of $H$ such that the $H_j$ are non-isomorphic and occur $\rho_j\geq 1$ times.
We can find the $H_j$ and $\rho_j$ in time $O(m^3\Delta^m)$ using Lemma~\ref{lem:connected count} (i) to test for isomorphism.
Write $\underline{H} = (H_1, \ldots, H_r)$, $h = (h_1, \ldots, h_r)$ with $h_i = |V(H_i)|$, and $\rho = (\rho_1, \ldots, \rho_r)$. Note that $\rho \cdot h = m$. Also recall 
 \[
 \mathcal{P}_{m,r,h} = \{ \gamma \in \mathbb{N}^r: \gamma \cdot h = m \},
 \]
and enumerate the elements of $\mathcal{P}_{m,r,h}$ as $\gamma_1, \ldots, \gamma_k$. 
We may assume that $\rho=\gamma_1$.
Given a vector $\nu \in \mathbb{N}^r$  let us write, as before, $\nu^*$ for the vector $(\nu^{\gamma_i \circ h})_{i=1}^k$. 
By Lemma~\ref{le:invertsys}, we can find vectors $\nu_1, \ldots, \nu_k \in \mathbb{N}^r$ such that the vectors $\nu_1^*, \ldots, \nu_k^* \in \mathbb{N}^k$ are linearly independent. 
We can find these vectors in time  $O(k^5 + k^2m e^{m}) = \tilde{O}(2^{10m})$, noting that by~\eqref{eq:bound P} $k$ is at most $\binom{m+r-1}{r-1} = O(2^{m+r}) = O(2^{2m})$.

Now, for $i= 1, \ldots, k$, consider the univariate polynomials $Z_{\nu_i}(z) = Z_{\nu_i, \underline{H}, G}(z)$. In particular, using Proposition~\ref{pr:BIGCP} and Theorem~\ref{thm:compute coef} we can compute the coefficient $s_m(\nu_i)$ of $z^m$ in $Z_{\nu_i}(z)$ in time $\tilde{O}(n (e\Delta)^{\alpha m}\beta_m4^{\alpha m})
$  with $\alpha = 1$ and $\beta_i = i^2r \Delta^{i-1}$. So computing al these coefficients can be done in time 
\[\tilde{O}(k\cdot (n(4\Delta)^{m})(e\Delta)^{m})=\tilde{O}(n(7\Delta)^{2m}).
\] 
Recall that this coefficient is 
\begin{equation*}
s_m(\nu_i)(G) = \sum_{j=1}^k \nu_i^{\gamma_j \circ h}s_{\gamma_j} = \sum_{j=1}^k \nu_i^{\gamma_j \circ h} \ind(\gamma_j H,G).
\end{equation*} 
More conveniently, writing $s \in\Z_{\geq 0}^k$ for the vector given by $s_j:=s_{\gamma_j}=(\ind(\gamma_jH, G))_{j=1}^k$, we have the invertible system of linear equations given by 
\[\nu_i^* \cdot s = s_m(\nu_i) \text{ for } i=1,\ldots,k,\]
where we have computed the values of $s_m(\nu_i)$ and $\nu_i^*$, while the vector $s$ is unknown (the system is invertible because we chose the $\nu_i^*$ to be linearly independent). 
We can then invert the system in time $O(k^3) = \tilde{O}(2^{6m})$. In particular finding the value of $s_1 = \ind(\rho H, G)$ can be done in $\tilde{O}(2^{6m})$.
The total running time is bounded by $\tilde{O}(n(7\Delta)^{2m}+2^{10m}))$.
\end{proof}

\section{Concluding remarks}\label{sec:conclusion}
As we remarked in the introduction our approach also works in the setting of vertex- and edge-coloured graphs.
We will not elaborate on the details here, but just refer the interested reader to Section 3.3 of \cite{PR16} where we have briefly explained how to extend the results for computing coefficients of BIGCPs to the setting of coloured graphs. 
In addition we note that the part of the proof given in Section~\ref{sec:proof main} also carries over to the coloured graphs setting replacing graph by coloured graph everywhere.

We moreover remark that the approach used to prove Theorem~\ref{thm:main} is very robust. 
Besides extending to the coloured setting, it also easily extends to other graph like structures. 
For example, in \cite{PR16} it has been extended to \emph{fragments}, i.e., vertex-coloured graphs in which some edges may be unfinished and more recently, Liu, Sinclair, and Srivastava \cite{LSS17} extended it to \emph{insects}, i.e., vertex-coloured hypergraphs in which some edges may be unfinished.
We expect our approach to be applicable to the problem of counting (induced) substructures in other structures as well, as long as there is a notion of connectedness and maximum degree.

\section*{Acknowledgements}
We thank John Lapinskas for raising a question about the complexity of our main algorithm in a previous version of this paper, which led to an improved running time. We also thank Radu Curticapean for informing us of some historical context to our result.

\bibliographystyle{abbrv}
\bibliography{approx}

\end{document}